\documentclass[conference]{IEEEtran}
\IEEEoverridecommandlockouts                              %

\usepackage{cite}
\usepackage{amsmath,amssymb,amsfonts,setspace}
\usepackage{algorithmic}
\usepackage{graphicx}
\usepackage{textcomp}
\usepackage{xcolor}

\newcommand{\reals}{\mathbb{R}}

\newcommand{\nullset}{\varnothing}

\newcommand{\norm}[1]{\left\lVert#1\right\rVert}

\DeclareMathOperator*{\argmax}{arg\,max}

\newtheorem{assume}{Assumption}
\newtheorem{rem}{Remark}

\newtheorem{lemma}{Lemma}
\newtheorem{pf}{Proof}

\let\labelindent\relax 

\usepackage{tikz, graphicx, mathtools, amsfonts, mathrsfs, amsmath, xfrac, dsfont, enumitem,multirow }

\usepackage[ruled,vlined]{algorithm2e}
\SetKwInOut{Parameters}{Parameters}
\SetKwInOut{Initialization}{Initialization}
\usepackage{orcidlink}
\newcommand{\nerf}{NeRF }
\newcommand{\nerfs}{NeRFs }

\begin{document}

\title{\LARGE \bf
Safe Navigation using Neural Radiance Fields via Reachable Sets
}

\author{\IEEEauthorblockN{Omanshu Thapliyal \orcidlink{0000-0003-3847-188X}}
\IEEEauthorblockA{\textit{Researcher, Hitachi America Ltd.}\\
Santa Clara, USA }
\and
\IEEEauthorblockN{Malarvizhi Sankaranarayanasamy}
\IEEEauthorblockA{\textit{Sr. Researcher, Hitachi America Ltd.}\\
Santa Clara, USA 
\\
\texttt{Malar.Samy@hal.hitachi.com}}
\and
\IEEEauthorblockN{Ravigopal Vennelakanti}
\IEEEauthorblockA{\textit{VP, Hitachi America Ltd.}\\
Santa Clara, USA 
\\
\texttt{Ravigopal.Vennelakanti}\\\texttt{@hal.hitachi.com}}}

\maketitle

\vspace{-1cm}
\begin{abstract}

Safe navigation in cluttered environments is an important challenge for autonomous systems.
Robots navigating through obstacle ridden scenarios need to be able to navigate safely in the presence of obstacles, goals, and ego objects of varying geometries.
In this work, reachable set representations of the robot's real-time capabilities in the state space can be utilized to capture safe navigation requirements.
While neural radiance fields (NeRFs) are utilized to compute, store, and manipulate the volumetric representations of the obstacles, or ego vehicle, as needed.
Constrained optimal control is employed to represent the resulting path planning problem, involving linear matrix inequality constraints.
We present simulation results for path planning in the presence of numerous obstacles in two different scenarios.
Safe navigation is demonstrated through using reachable sets in the corresponding constrained optimal control problems.
\end{abstract}

\section{Introduction}

Neural radiance fields (NeRFs) have emerged as a useful technique for complex visual modeling and computer graphics.
In the eponymous work in \cite{mildenhall2021nerf}, \nerfs were found to outperform existing convolutional neural network-based methods in discretized voxel representation of 3D objects.
The volumetric representational capabilities of \nerfs have been studied in depth for high-fidelity scene rendering from limited training data comprising of multi-pose images and depth data taken from simple camera sensors.
Due to their training via raymarching, \nerfs can be provided with sparse inputs, and generate realistic volumetric representations with smooth and continuous contours of the given object.

Since \cite{mildenhall2021nerf}, \nerfs have been applied to various engineering problems \cite{yu2021plenoctrees, gao2022nerf, li20213d}, including scene reconstruction, view synthesis, object tracking \& pose estimation, biomedical imaging, and autonomous systems.
Particularly for us, autonomous robot applications of \nerfs include path planning in visual renderings of scenarios \cite{adamkiewicz2022vision}.
Authors in \cite{kwon2023renderable} suggest similar rendering-based path planning methods where the entire 3D scene is represented via detailed pre-trained \nerf renderings of the complete scenario.
The autonomous robot then attempts to navigate through the \nerf scene.
The resulting volumetric representations studied in literature can allow for spatial reasoning towards safe locomotion of autonomous robots, which is the scope of this work.

Similarly, state space approaches to safety characterization via \textit{reachability} properties have been widely studied in control systems literature \cite{bansal2017hamilton,chen2018hamilton,chen2017reachability, thapliyal2022approximate, thapliyal2023approximating, thapliyal2024embedding}.
State space reachability of a dynamical system involves computing all possible states a dynamical system can take under a given control input set, within a given time horizon.
As a result, (forward) reachability is often utilized to predict future safe evolutions of trajectory sets, starting from given initial sets/conditions.
``Lookahead"-safety via purely data-driven methods utilizing reachability combined with neural networks are considered in \cite{thapliyal2022approximate} (in the form of deep neural networks), and \cite{thapliyal2024embedding} (in the form of physics-informed neural networks).
Similar methods for machine learning approaches to learn reachable sets to determine safety properties of dynamical systems are discussed in \cite{bansal2021deepreach}.
Safe autonomy has been addressed in numerous texts and methods -- reachable sets being an important volumetric representation of robot's real-time capabilities in the state space.
Therefore, reachability based control and verification finds applications in aerospace engineering, automotive control, robotics, and cyberphysical systems \cite{bansal2017hamilton, chen2018hamilton, michaux2023reachability}.

Motivated by the utility of reachability in expressing state-space safety constraints and system properties, and by the representational power of \nerfs for complex 3D geometries, this work explores the following.
Safe navigation of autonomous systems with low-power sensors can be enhanced by using \nerf-based volumetric models of robots and obstacles.
Reachability-based safety constraints then guide motion in cluttered environments.
Since convex hulls of \nerf geometries and reachable sets yield convex constraints, the resulting safe path-planning problem can be solved efficiently.
The integration of \nerf representations with reachable sets and the corresponding path-planning formulation constitute the main contributions of this work.

The rest of this paper is organized as follows.
Section \ref{sec:prob} outlines the problem formulation for safe navigation using \nerfs\ and their geometric representations for safe path planning via reachability. 
Section \ref{sec:method} presents the proposed solution. Section \ref{sec:example} demonstrates two path planning scenarios with \nerf-based obstacles and a robot vehicle. 
Finally, Section \ref{sec:conclusion} provides concluding remarks and future directions.

\section{Problem Formulation}\label{sec:prob}

Consider a mobile robot with linear dynamics given by:
\begin{equation}\label{eq:ltidynamics}
\dot{x}(t)=A(t)x(t)+B(t)u(t)\,\end{equation}
where the robot state at time $t\geq 0$ is given by $x\in\reals^n$, control input at time $t$ is given by $u\in\mathcal U\subset\reals^m$, and the time-varying system matrices are given by $A\in\reals^{n\times n}$, and $B\in\reals^{n\times m}$, respectively.
Further, let the initial state lie in some set $x(0)\in\mathcal X_0$. 
The path planning problem for the mobile robot includes a scenario description, defined as follows.
Let $\Omega\subset\reals^n$ be the path planning arena, a bounded region of the state space.
Let $\mathcal O=\{O_1,O_2,\cdots,O_{N}\}$ where $O_i\subset \Omega$ be the static obstacles populating the scenario, and $\mathcal G\subset\Omega$ be the goal state set.
Finally, let $\mathcal S\subset \Omega$ be the robot geometric representation, denoting the physical boundaries of the mobile robot.
Therefore, the path planning problem for some planning horizon $T$ is:
\begin{equation}\label{eq:pathplanprob}
\begin{split}
\text{find }\; & u(t) \\
\text{ s.t. }\; & x(\tau) \in\mathcal G, \tau\leq T, \\
& \mathcal S(\tau) \cap \mathcal O_i = \nullset \;\forall i \in \{1,\cdots,N\}, \forall \tau\leq T \\
\text{given } & \Gamma=\{\{(A(t),B(t))\}_{t\leq T}, \Omega,\mathcal O,\mathcal G,\mathcal S, \mathcal X_0\}
\end{split}
\end{equation}

Clearly, a safe navigation policy, $u^*(t)$ that solves (\ref{eq:pathplanprob}) models a safe mobile robot trajectory, if such a control sequence exists.
In the next section, we outline a reachability-based anticipating control methodology for the above problem.

\begin{figure*}[!htp]
    \centering
    \resizebox{0.65\textwidth}{!}
    {
  
\tikzset {_6tgfyomj9/.code = {\pgfsetadditionalshadetransform{ \pgftransformshift{\pgfpoint{89.1 bp } { -108.9 bp }  }  \pgftransformscale{1.32 }  }}}
\pgfdeclareradialshading{_8z8rqnw52}{\pgfpoint{-72bp}{88bp}}{rgb(0bp)=(1,1,1);
rgb(0bp)=(1,1,1);
rgb(25bp)=(0,0,0);
rgb(400bp)=(0,0,0)}

  
\tikzset {_jfzmwe045/.code = {\pgfsetadditionalshadetransform{ \pgftransformshift{\pgfpoint{89.1 bp } { -128.7 bp }  }  \pgftransformscale{1.32 }  }}}
\pgfdeclareradialshading{_2cswvzbq1}{\pgfpoint{-72bp}{104bp}}{rgb(0bp)=(1,1,1);
rgb(0bp)=(1,1,1);
rgb(19.196428571428573bp)=(0.48,0.15,0.15);
rgb(400bp)=(0.48,0.15,0.15)}

  
\tikzset {_2vh0k6jjw/.code = {\pgfsetadditionalshadetransform{ \pgftransformshift{\pgfpoint{89.1 bp } { -108.9 bp }  }  \pgftransformscale{1.32 }  }}}
\pgfdeclareradialshading{_h6p5guust}{\pgfpoint{-72bp}{88bp}}{rgb(0bp)=(1,1,1);
rgb(0bp)=(1,1,1);
rgb(25bp)=(0,0,0);
rgb(400bp)=(0,0,0)}
\tikzset{every picture/.style={line width=0.75pt}} 

\begin{tikzpicture}[x=0.75pt,y=0.75pt,yscale=-1,xscale=1]

\draw  [draw opacity=0][shading=_8z8rqnw52,_6tgfyomj9] (85.33,267.6) .. controls (85.33,249.37) and (106.6,234.6) .. (132.83,234.6) .. controls (159.07,234.6) and (180.33,249.37) .. (180.33,267.6) .. controls (180.33,285.83) and (159.07,300.6) .. (132.83,300.6) .. controls (106.6,300.6) and (85.33,285.83) .. (85.33,267.6) -- cycle ;
\draw  [draw opacity=0][shading=_2cswvzbq1,_jfzmwe045] (242.67,140.27) .. controls (242.67,113.94) and (273.48,92.6) .. (311.5,92.6) .. controls (349.52,92.6) and (380.33,113.94) .. (380.33,140.27) .. controls (380.33,166.59) and (349.52,187.93) .. (311.5,187.93) .. controls (273.48,187.93) and (242.67,166.59) .. (242.67,140.27) -- cycle ;
\draw  [fill={rgb, 255:red, 155; green, 155; blue, 155 }  ,fill opacity=0.39 ][dash pattern={on 4.5pt off 4.5pt}] (126.71,182.38) -- (175.37,214.35) -- (71.98,326.3) -- (23.32,294.33) -- cycle ;
\draw    (97.67,257.27) -- (132.83,267.6) ;
\draw [shift={(97.67,257.27)}, rotate = 16.37] [color={rgb, 255:red, 0; green, 0; blue, 0 }  ][fill={rgb, 255:red, 0; green, 0; blue, 0 }  ][line width=0.75]      (0, 0) circle [x radius= 3.35, y radius= 3.35]   ;
\draw  [fill={rgb, 255:red, 248; green, 231; blue, 28 }  ,fill opacity=0.2 ][dash pattern={on 4.5pt off 4.5pt}] (236.66,49.32) -- (306.09,65.74) -- (245.42,189.36) -- (176,172.93) -- cycle ;
\draw    (257,123.93) -- (311.5,140.27) ;
\draw [shift={(257,123.93)}, rotate = 16.68] [color={rgb, 255:red, 0; green, 0; blue, 0 }  ][fill={rgb, 255:red, 0; green, 0; blue, 0 }  ][line width=0.75]      (0, 0) circle [x radius= 3.35, y radius= 3.35]   ;
\draw [line width=1.5]  [dash pattern={on 5.63pt off 4.5pt}]  (99.35,254.34) .. controls (213,231.93) and (213.67,153.27) .. (253.67,123.27) ;
\draw [line width=1.5]  [dash pattern={on 1.69pt off 2.76pt}]  (132.83,267.6) .. controls (311.67,242.6) and (226.33,198.6) .. (307.67,140.6) ;
\draw  [draw opacity=0][shading=_h6p5guust,_2vh0k6jjw] (477.7,221.25) .. controls (471.45,198.41) and (495.14,172.03) .. (530.62,162.32) .. controls (566.09,152.61) and (599.92,163.25) .. (606.17,186.08) .. controls (612.42,208.91) and (588.73,235.29) .. (553.25,245) .. controls (517.77,254.71) and (483.95,244.08) .. (477.7,221.25) -- cycle ;
\draw   (557.39,173.78) -- (501.83,202.59) -- (478.67,191.58) -- (519.92,155.98) -- (568.57,144.98) -- cycle ;
\draw   (501.83,202.59) -- (557.4,173.78) -- (562.95,215.61) -- (534.47,249.08) -- (491.8,247.07) -- cycle ;
\draw   (534.47,249.08) -- (562.95,215.61) -- (579.37,206.61) -- (575,246.07) -- (552.6,264.17) -- cycle ;
\draw   (475.8,229.98) -- (475.8,229.98) -- (478.66,191.58) -- (501.83,202.59) -- (491.8,247.07) -- cycle ;
\draw   (520.6,270.2) -- (491.8,247.07) -- (534.47,249.08) -- (534.47,249.08) -- (552.6,264.17) -- cycle ;
\draw   (593.13,167.63) -- (568.58,144.98) -- (593.13,167.63) -- (593.13,167.63) -- (579.37,206.61) -- cycle ;
\draw   (619.8,192.77) -- (593.13,167.63) -- (579.37,206.61) -- (575,246.07) -- (619.8,192.77) -- cycle ;

\draw  [line width=2.25]  (446.4,95.8) -- (634.2,95.8) -- (634.2,341.4) -- (446.4,341.4) -- cycle ;

\draw (132.83,300.6) node [anchor=north west][inner sep=0.75pt]  [font=\large]   {$\mathcal{X} _{0}$};
\draw (360.67,14.93) node [anchor=north west][inner sep=0.75pt]  [font=\large]   {$\mathcal{R}^{z }( \tau ;\mathcal{X}_{0})$};
\draw (66,243.6) node [anchor=north west][inner sep=0.75pt]  [font=\large]   {$z _{0j}^{*}$};
\draw (48.67,166.93) node [anchor=north west][inner sep=0.75pt]  [font=\large]   {$\left( c_{0j} ,\gamma _{0j}^{*}\right)$};
\draw (30,129.6) node [anchor=north west][inner sep=0.75pt]  [font=\large]   {$j^{th} \ \mathsf{Hyperplane}$};
\draw (94.67,327.33) node [anchor=north west][inner sep=0.75pt]  [font=\large]   {$\mathsf{Initial\ Set}$};
\draw (360.67,61.6) node [anchor=north west][inner sep=0.75pt]  [font=\large]   {$ \begin{array}{l}
\mathsf{Reachable\ Set}\\
\mathsf{at\ the\ \tau }
\end{array}$};
\draw (204.67,106.27) node [anchor=north west][inner sep=0.75pt]  [font=\large]   {$z _{j}^{*}( \tau)$};
\draw (244.33,240.93) node [anchor=north west][inner sep=0.75pt]  [font=\large]   {$\dot{z }( t) ={A}(t) z ( t) \ +B(t) u( t)$};
\draw (473.87,267.33) node [anchor=north west][inner sep=0.75pt]  [font=\large]   {$\bigcap _{j=1}^{n_{p}}\left\{z \in \left( c_{0j} ,\gamma _{0j}^{*}\right)\right\}$};
\draw (452.27,101.73) node [anchor=north west][inner sep=0.75pt]  [font=\large]   {$\mathsf{Polytopic\ Approximation}$};

\end{tikzpicture}
    }
    \caption{Polytopic Approximation of Reachable Sets: Evolution of a selected Hyperplane under dynamics $(A(t),B(t))$}
    \label{fig:polytope}
\end{figure*}
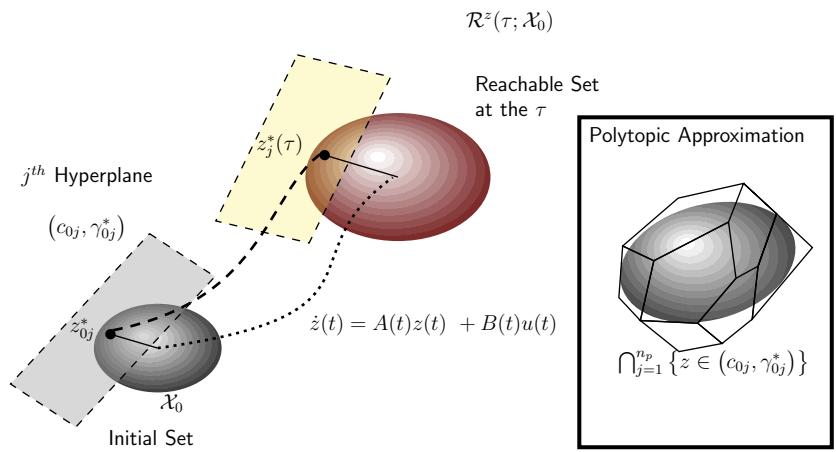

\section{Safe Navigation Methodology}\label{sec:method}
For the path planning problem in (\ref{eq:pathplanprob}), the target pursuing and obstacle avoiding constraints pose biggest challenges to numerical solutions of the associated optimal control problem.
As noted in our previous work \cite{thapliyal2021path}, state-space based constraints due to obstacle avoidance goals of the type $S(\tau)\cap \mathcal O_i=\nullset$ in (\ref{eq:pathplanprob}) result in non-convex feasible regions.
A natural consequence of this non-convexity and improper characterization of the feasible region of the state space where the robot can locomote is difficulty in numerical optimization.
To alleviate this problem, a shorter ``lookahead"-based path planning can be performed to replace the set intersection-based constraints in (\ref{eq:pathplanprob}) by employing reachable sets of the mobile robot with given dynamics.

A \textit{$T$-time reachable set} of the dynamical system $(A(t),B(t))$ in (\ref{eq:ltidynamics}), with admissible control input space $\mathcal U$, and initial condition set $\mathcal X_0$ is simply the set of all possible state space points that can be reached by the system.
In other words, such a set is defined as:
\begin{equation}\label{eq:reachdefn}
\begin{split}
\mathcal R (T;\mathcal X_0,\mathcal U) &\triangleq \left\{z(T)\in\reals^n:\dot{z}=A(s)z(s)+B(s)u(s),\right.\\
& \left. s\leq T, u\in\mathcal U, z(0)\in\mathcal X_0\right\}    
\end{split}
\end{equation}
To compute the reachable set in (\ref{eq:reachdefn}), we can invoke a method based on optimal control adapted from \cite{thapliyal2022approximate} as follows.
To this end, we begin with an assumption on the initial set $\mathcal X_0$ and the admissible control set $\mathcal U$.
\begin{assume}
Sets $\mathcal S, \mathcal X_0$ and $\mathcal O_i$ for $i=1,\cdots,N$ can be bounded by closed polytopes.
\end{assume}
Note that the assumption above is not restrictive as most path planning problems deal with bounded scenario sets $\Omega$.
An example of the geometric representation of the bounded state space set $\mathcal X_0$, and the resulting propagation of the set $\mathcal X_0$ in time $\tau$ under dynamics $(A(t),B(t))$ is shown in Fig.~\ref{fig:polytope}.

From Fig.~\ref{fig:polytope}, suppose the $n_p$-sided polytope bounding the initial set is given by an intersection of $n_h$ number of half planes as:
\begin{equation}\label{eq:half-plane}
\mathcal X_0 = \{z\in\reals^n:\langle c_{0j}, z\rangle \leq \gamma^*_{0j} \}    
\end{equation}
where the contact point of the hyperplane parameterized by $(c_{0j},\gamma^*_{0j})$ with the set $\mathcal X_0$ is the point $z^*_{0j}\in\reals^n$ (for $j=1,\cdots,n_p$).
Such a singular hyperplane is depicted in gray in Fig.~\ref{fig:polytope}.
The polytope resulting from the intersection of all the half-planes is shown in  Fig.~\ref{fig:polytope}, in the inset to the right.

With the polytopic definition of $\mathcal X_0$ under Assumption 1, and bounded admissible control set $\mathcal U$, the resulting reachable set $\mathcal R(T)$ defined in (\ref{eq:reachdefn}) is known to also be polytopic \cite{varaiya2000reach} under linear dynamics.
Furthermore, the contact points $z^*_{0j}$`s can be propagated in time to estimate the $T$ seconds in the future.
The propagated points go under some optimal control law $u^*(\cdot)$ that attempts to push the point $z^*{0j}$ furthest away from the hyperplane, thereby resulting in an extreme-point optimal control problem.
The geometric intuition thus described is also shown in Fig.~\ref{fig:polytope}, where the gray hyperplane (and its point of contact $z_{0j}$) are propagated at some time $t$ to the yellow hyperplane (and its point of contact $z^*_j(t)$).
This allows us to compute (and propagate) the polytopic reachable sets in real-time for linear systems.
This well known result has been noted in our previous works \cite{thapliyal2022approximate,thapliyal2023approximating}, and was initially known due to \cite{varaiya2000reach}.
Here we restate the polytopic reachable set computation method, as required by the problem posed in (\ref{eq:reachdefn}).

\begin{lemma}
Suppose hyperplane $(c_{0j},\gamma^*_{0j})$ supports the set $\mathcal X_0$ at the point $z_{0j}^*$, then the point evolves in time as:
\begin{align}\label{eq:support-point-evolve}
\dot{z}_j(\tau) &= A(t) z_j(\tau) + B(t) u_j^*(\tau)
\end{align}
where $u_j^*(\tau)$ solves the following optimal control problem:
\begin{align}
u_j^*(\tau) &= \underset{u\in\mathcal U}{\argmax} {\left\{ \langle c_j^*(\tau), A(\tau) z_j^*(\tau) + B(\tau) u(\tau)\rangle \right\}} \label{eq:optimal-control} \\    
\dot{c}_j^*(\tau) &= -A(\tau)^T{c}_j^*(\tau)\text{ s.t. }, {c}_j^*(t_0) = c_j(t_0), \,t_0\leq\tau\leq t \label{eq:costate-evolve}
\end{align}
for the $j^\text{th}$ hyperplane, and the co-state variable $\lambda_j^*(\tau)$.
\end{lemma}
\begin{pf}
The proof follows from Varaiya's polytopic reachability results in \cite{varaiya2000reach} and Theorem 1 in \cite{thapliyal2023approximating}.    
\end{pf}

Note that the results from Lemma 1 allow a closed form computation of reachable set $\mathcal R(T;\mathcal X_0,\mathcal U)$ over continuous time.
The computational complexity grows with the number of hyperplanes $n_p$ required to approximate the reachable set $\mathcal R(T)$.
While safe obstacle avoidance is guaranteed for the next $T$ seconds for the point vehicle, if $\mathcal R(T)\cap \mathcal O_i=\nullset$ for each obstacle $i$, such a guarantee holds point-wise along the trajectory.
The polytopic representation of $\mathcal R(T)$ allows us to simplify the path planning constraints in the original problem $(\ref{eq:pathplanprob})$, upon combining with the ego robot geometry.

To this end, we utilize neural radiance fields \nerf to represent the obstacles.
This is done, in practice, via mono cameras on-board the vehicle, or offline to store 3D geometries of the objects in the scenario $\Gamma$, as shown in Fig.~\ref{fig:nerf-description}.

\begin{figure}[!ht]
    \centering
    \includegraphics[width=0.85\columnwidth]{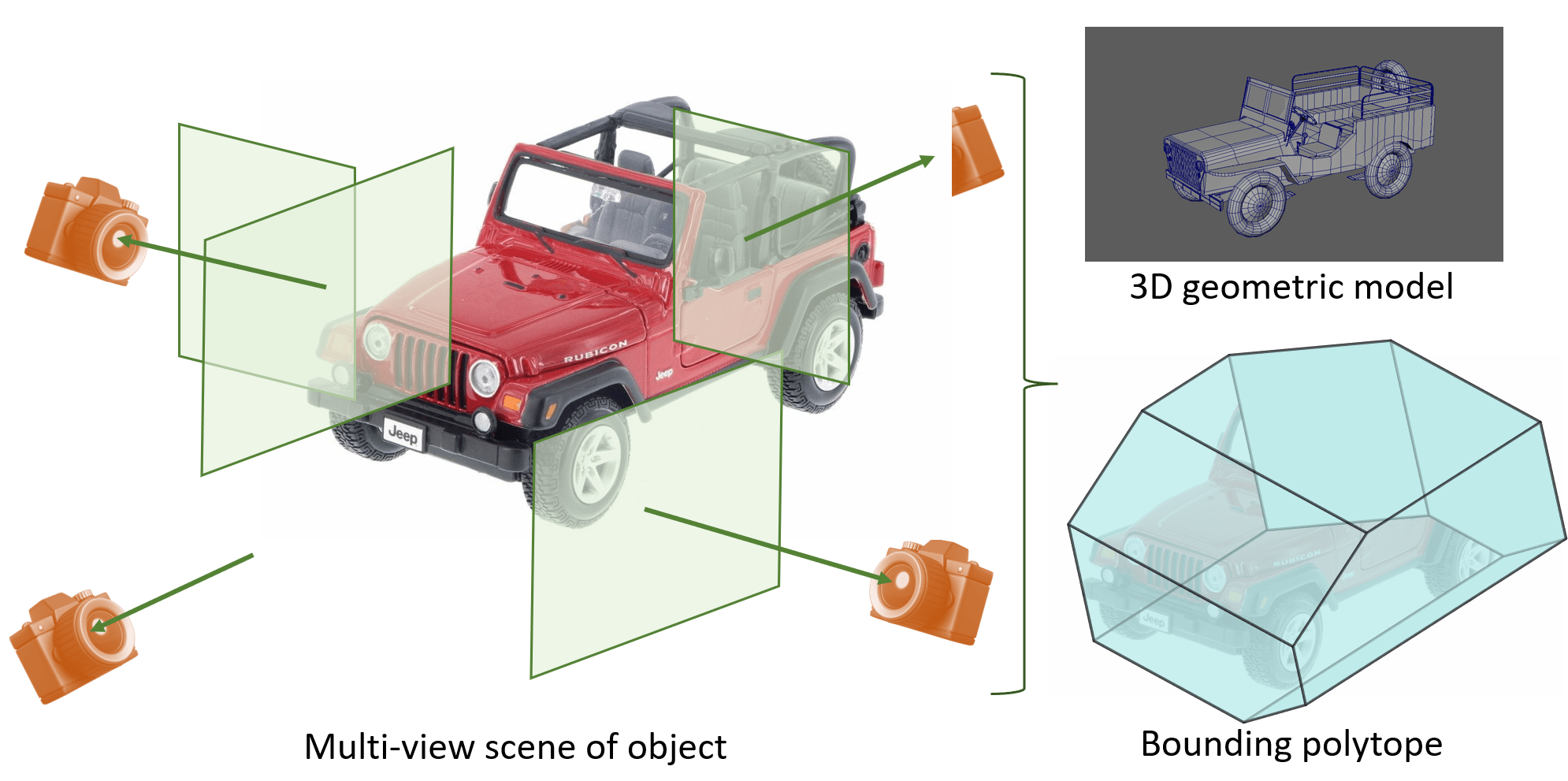}
    \caption{Reconstructing a 3D object using \nerfs}
    \label{fig:nerf-description}
\end{figure}
Using image data recorded over varying viewing directions as shown in Fig.~\ref{fig:nerf-description}, a \nerf learns a neural radiance function at the given viewing direction for the object.
In this case, the input data to the \nerf is a sequence of images across different camera angles $I_i\in\reals^3\times\reals^2$, where $I_i={x,y,z,\theta,\phi}$ represents the 3D point and viewing angle.
The neural radiance output function $F:\reals^3\times\reals^2\to\reals^4$ maps to RGB color and intensity. 

The core \nerf architecture is a feedforward neural network with inputs from $\reals^3\times\reals^2$. 
Training uses image sequences ${I_i}$ to minimize predicted versus observed radiance \cite{mildenhall2021nerf}. 
Assuming a continuous volumetric scene, \nerfs produce high-fidelity 3D renderings suitable for generating point clouds and corresponding convex hull polytopes.

No distinction is made between \nerf\ objects as obstacles, goals, or the ego vehicle. 
Using hull polytopes derived from their point clouds, we can now perform path planning for scenario $\Gamma$.
We first define the Minkowski sum of two sets $\mathbb S$ and $\mathbb T$ denoted by the operator $\oplus$, defined as $\mathbb S\oplus \mathbb T\triangleq \{s+t:s\in\mathbb S, t\in\mathbb T\}$.
\begin{figure}[!ht]
    \centering
    \includegraphics[width=0.7\columnwidth]{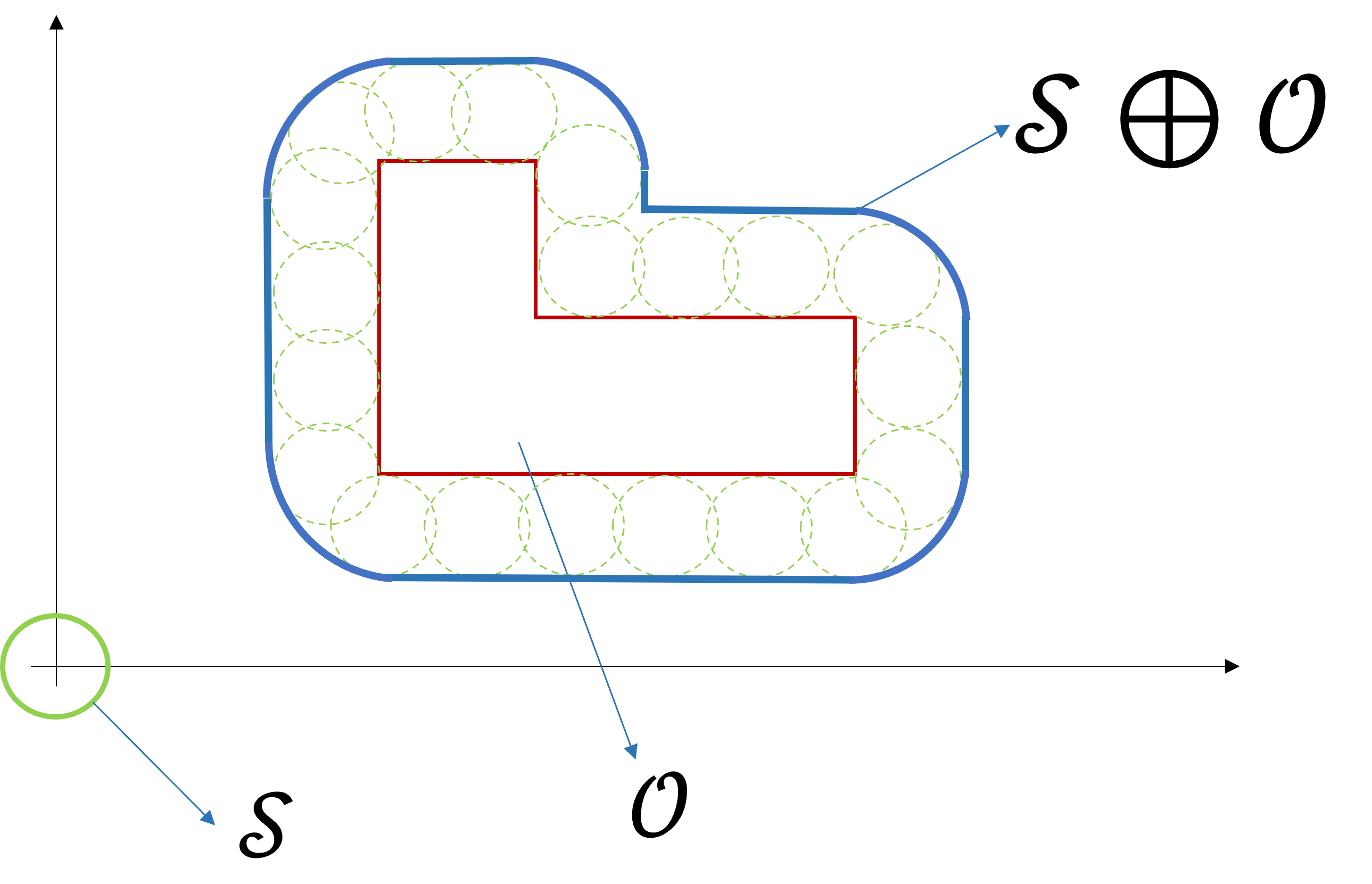}
    \vspace{-0.2cm}
    \caption{Minkowski sum of robot geometry with a polygonal obstacle}
    \vspace{-0.2cm}
    \label{fig:config-space}
\end{figure}
The collection of all possible configurations of the robot around the obstacle $\mathcal O$, called the configuration space ($\mathcal C$-space), is shown in Fig.~\ref{fig:config-space}.
Given the robot and obstacle shapes, the configuration space can be computed by their Minkowski sums.
Minkowski sum-based methods for path planning have been utilized in literature (see \cite{6284872,ziegler2010fast,aronov1994translational} for more details) at a great success.

Minkowski sum operation effectively expands the robot shape by the shape of each obstacle, taking into account all possible  configurations of the robot and all positions of the obstacle.
So far, no anticipatory information was injected to the configuration space, that is, all possible relative configurations of the obstacle-robot pair are taken at a given time instance, say $t$.
However, path planning for some $T$ seconds in the future requires a `sweeping operation' across time $t\leq T$ of the relative configurations.
To this end, we can utilize the polytopic representations of the robot object, obstacles, and the reachable set to get rid of the set-based constraints in (\ref{eq:pathplanprob}).

\begin{algorithm}[!t]
\setstretch{1.0}
\SetNoFillComment
\SetAlgoLined
\KwIn{$\Gamma=\{\{(A(t),B(t))\}_{t\leq T}, \Omega,\mathcal O,\mathcal G,\mathcal S, \mathcal X_0\}$}
\Parameters{\nerf description of ego robot, obstacles, goal set\\
\nerf object representation:\\
Polytope {$Fx\leq b$} for $\mathcal S$,\\
Polytope {$Gx\leq c$} for $\mathcal G$,\\
Polytope {$H_ix\leq d_i$} for $\mathcal O_i$\\}
Compute $\mathcal C$-space $\leftarrow$ \texttt{minkowskiSum($\alpha,\beta$)} \cite{aronov1994translational}\\
\While{$\tau\leq T$}
{
    compute $\mathcal R(\tau)\leftarrow$ Lemma 1\\
    model predictive step:\\
    min $d(\text{polytope}(G,c), \mathcal R(\tau))$\\
    subject to:\\
    $\mathcal R(\tau) \cap$ {polytope}($H_i,d_i$) $=\nullset$\\
    $\dot{x}(\tau)=A(t)x(\tau)+B(t)u(\tau)$
} 
\caption{Safe path planning via \nerfs and reachable sets}\label{algo:algo1}    
\end{algorithm}

\begin{rem}
The swept region configuration space can be visualized along with the $T$-time reachable tube as an inflated obstacle region for the relative robot-obstacle geometric configurations.
This allows for including dynamical models implicitly into the polytopic regions used for path planning
\end{rem}
\begin{rem}
The computations in Algorithm~\ref{algo:algo1} are carried out for polytopes, as they are known to be efficiently Minkowski summable in literature \cite{ziegler2010fast,lien2008hybrid}.
\end{rem}

\section{Numerical Simulation}\label{sec:example}
Consider a mobile robot navigating through 3D state space in the presence of numerous static obstacles.
We consider double integrator dynamics for the mobile robot, given by:
\begin{equation}\label{eq:integrator}
m\ddot{x}(t)=F_x,\,m\ddot{y}(t)=F_x,\,m\ddot{z}(t)=F_z  
\end{equation}
This can be readily represented as an LTI system with the state $[x,\dot{x},y,\dot{y},z,\dot{z}]^T]$, with control input $[F_x/m, F_y/m, F_z/m]^T$.
The robot geometry is defined by a sequence of camera views along varying $(\theta,\phi)$, as shown in Fig.~\ref{fig:cameraviews}.
\begin{figure}[!ht]
    \centering
    \includegraphics[width=0.75\columnwidth]{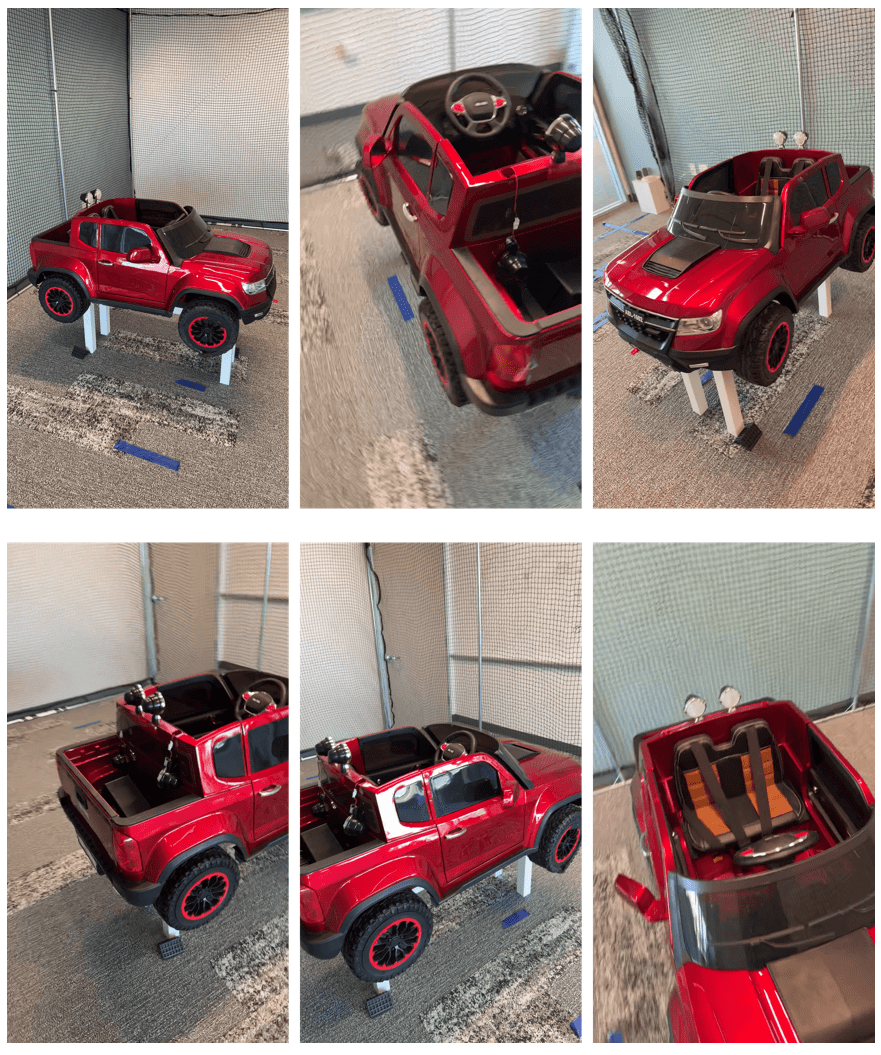}
    \caption{Training data for the \nerf along different camera views}
    \label{fig:cameraviews}
\end{figure}
To estimate the robot geometry, we trained \nerf models via \texttt{nerfstudio} (see \cite{tancik2023nerfstudio} for implementation details) on an Intel Core i9-12900HX, 2300 Mhz, 16 Core processor.
The training process took approximately 18 minutes, averaged across multiple trials, and a visualization of the 3D pointcloud of the resulting \nerf object is shown in Fig.~\ref{fig:pcd}.
\begin{figure}[!ht]
    \centering
    \vspace{-0.2cm}
    \includegraphics[width=0.8\columnwidth]{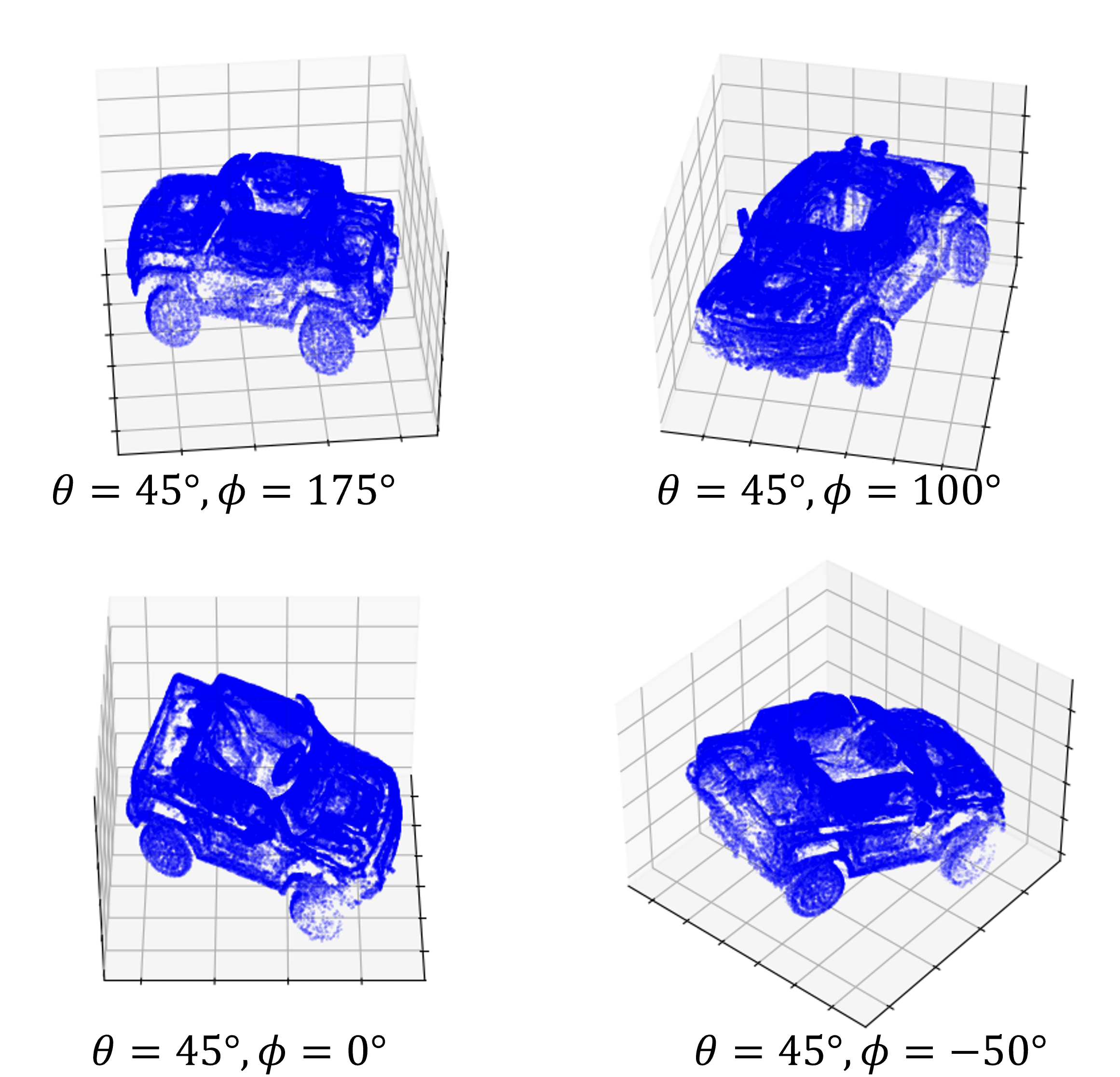}
    \caption{Trained \nerf object: 3D pointcloud}
    \label{fig:pcd}
\end{figure}
The convex hull of the resulting pointcloud is depicted in Fig.~\ref{fig:convexhull}.
Due to the large number of points and articulated surfaces of the object, the convex hull consists of 892 linear inequalities.
The convex hull can be represented as an equivalent system of inequalities as a linear matrix inequality $Fx\leq b$.
Similarly, the polytopic reachable set for $t=0.5$s is expressed as $Gx\leq c$.
\begin{figure}[!ht]
    \centering
    \vspace{-0.4cm}
    \includegraphics[width=0.9\columnwidth]{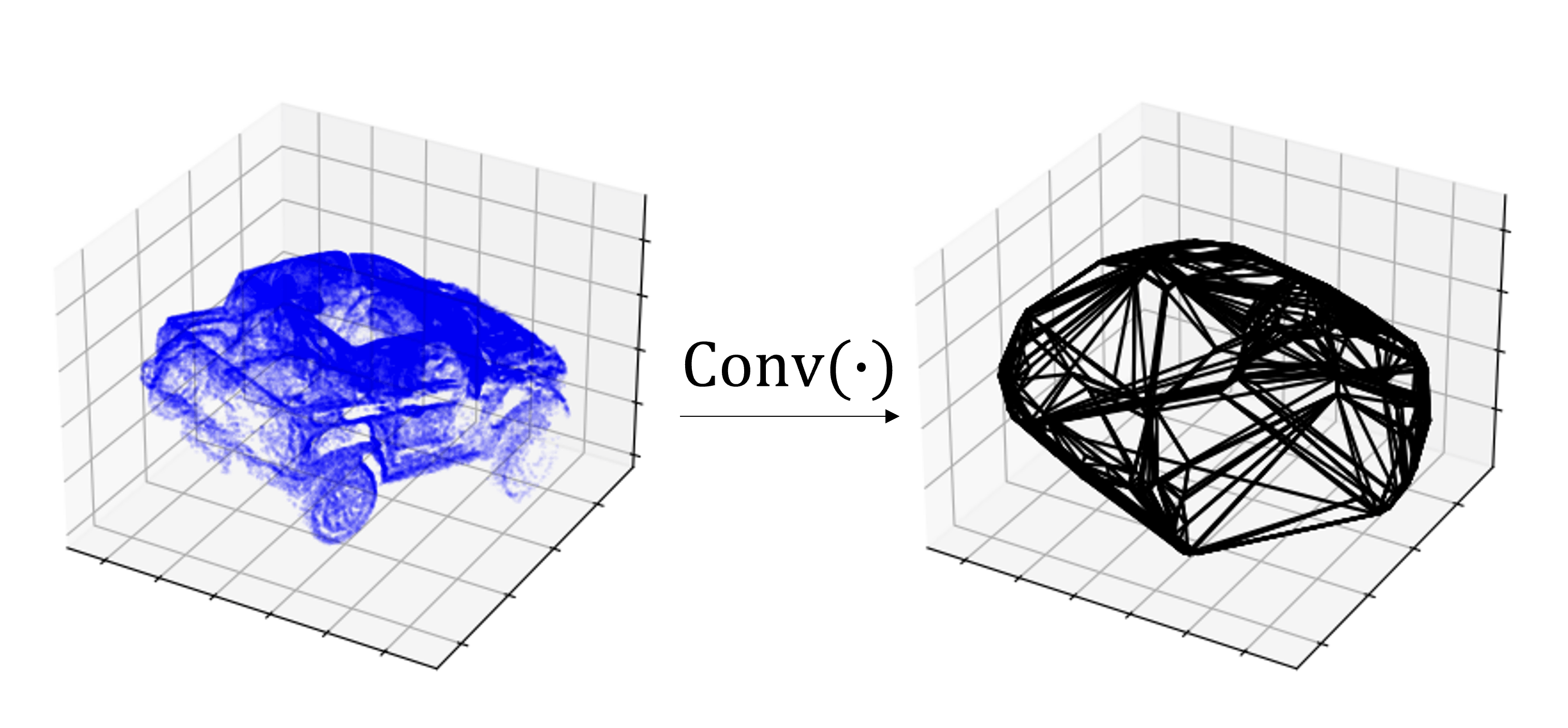}
    \vspace{-0.3cm}
    \caption{Convex hull of the trained \nerf object to describe robot geometry}
    \label{fig:convexhull}
\end{figure}

The arena $\Omega$ consists of a cubical region in 3D with a side length of 8 units.
The robot starts off at $(0,0,0)$, with the goal point at $(7.5,7.5,7.5)$, and 8 cuboidal obstacles of randomly generated sizes populating the arena.
Let $H_ix\leq d_i$ be the polytopic description of the $i$-th obstacle. 
This completely describes the arena $\Gamma$ in the simulation.
The path planning problem in (\ref{eq:pathplanprob}) can then be rewritten as:
\begin{equation}\label{eq:equivproblem1}
\begin{split}
\text{find }\; & u(t) \\
\text{ s.t. }\; & Fx(\tau) \leq b, H_i x(\tau)\geq d_i\; \forall i, \tau\leq T, \\
& \exists x(\tau) \in \mathcal G \text{ such that }G(\tau) x(\tau) \leq c, \tau\leq T\\
& \dot{x}(\tau) = A x(\tau) + B u(\tau)
\end{split}
\end{equation}
This way, the set-based constraints can be written as a proper optimal control problem with linear inequality constraints:
\begin{equation}\label{eq:equivproblem2}
\begin{split}
\text{minimize }\; & d(\mathcal G, \mathcal R(\tau)) \\
\text{ s.t. }\; & Fx(\tau) \leq b, H_i x(\tau)\geq d_i\; \forall i, \tau\leq T, \\
& \dot{x}(\tau) = A x(\tau) + B u(\tau)
\end{split}
\end{equation}
where $d(S_1,S_2)$ is the distance function between sets $S_1$ and $S_2$ defined as $\inf{\norm{x_1-x_2}}$ such that $x_1\in S_1$ and $x_2\in S_2$.

This optimal control problem in (\ref{eq:equivproblem2}) is solved via model predictive control, in a cluttered environment with 6 obstacles.
The resultant path planning trajectory with ego robot geometry represented by the convex hull in Fig.~\ref{fig:convexhull}, and the $T$-time reachable sets are used to perform obstacle avoidance, is shown in Fig.~\ref{fig:result}.

A similar simulation scenario is shown in Fig.~\ref{fig:result2}, with 10 randomly placed obstacles. Instead of the robot geometry, the goal state is given by the \nerf object.
The robot is able to navigate past multiple obstacles in both scenarios, safely clearing the obstacle walls while reaching the goal state.
\begin{figure}[!ht]
    \centering
    \vspace{-0.4cm}
    \includegraphics[width=0.85\columnwidth]{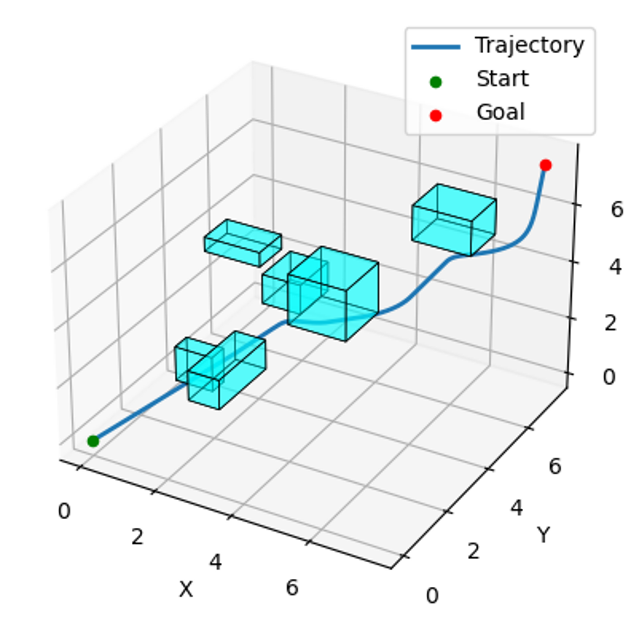}
    \vspace{-0.4cm}
    \caption{\nerf object represents robot geometry, and reachable sets to impose obstacle avoidance}
    \label{fig:result}
\end{figure}

\begin{figure}[!ht]
    \centering
    \vspace{-0.6cm}
    \includegraphics[width=1\columnwidth]{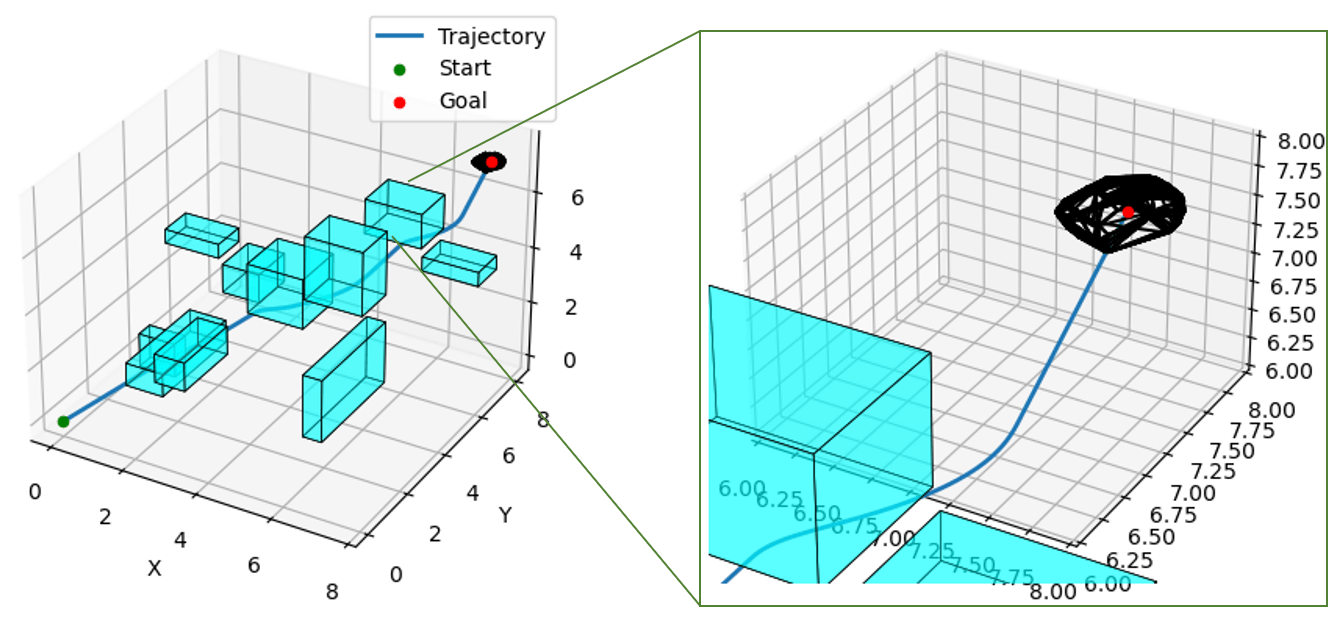}
        \vspace{-0.4cm}
    \caption{\nerf objects represents goal set, and reachable sets to impose obstacle avoidance (in the inset: \nerf object convex hull at goal state)}
    \label{fig:result2}
\end{figure}


\section{Conclusion}\label{sec:conclusion}
In this work we utilized reachable sets to impose constraints on path planning, where the obstacle and/or robot geometries are represented via neural radiance field (NeRF) objects.
The related reachable set computations was carried out in a computationally inexpensive manner via polytopic reachability for linear systems.
The convex hulls of the \nerf objects can be utilized to perform safe path planning, where safety constraints were imposed via reachable sets. 
Since the convex hulls of the \nerf objects can represent goal states, obstacles, robot geometry, and reachable sets are all represented as polytopes, the resulting optimal control problems involve linear matrix inequalities.
The resulting path planning scheme was demonstrated via numerical simulations of path planning in cluttered scenarios.
State-space safety representations stand to benefit significantly from the geometric properties of \nerf objects. 

Our immediate future work will apply this method to safe path planning for robotic manipulators and conduct comparative analyses of storage and computational efficiency relative to constrained optimal control approaches.

\bibliographystyle{IEEEtran}
\bibliography{bibliography}

\end{document}